\documentclass[11pt]{article}

\usepackage{amstext}
\usepackage{amssymb}
\usepackage{amsthm}
\usepackage{amssymb}
\usepackage{amsmath}
\usepackage{epsfig}
\usepackage{float}
\usepackage{times}
\usepackage{verbatim}
\usepackage{subcaption}

\DeclareMathOperator*{\argmax}{arg\,max}

\newcommand{\Pmatrix}[1]{\begin{pmatrix}#1\end{pmatrix}}
\newcommand{\theset}[1]{\left\{#1\right\}}
\newcommand{\magn}[1]{{\left| #1 \right|}}
\newcommand{\norm}[1]{\left\Vert #1 \right\Vert}
\newcommand{\st}{\textrm{\ \Big\vert\ }}
\newcommand{\tr}{^\mathrm{T}}

\newcommand{\und}{\ \&\ }
\newcommand{\Det}[1]{\mathbf{det}\left[#1\right]}
\newcommand{\Int}[1]{\mathbf{Int}\thinspace #1}
\newcommand{\diag}[1]{\mathbf{diag} #1}
\newcommand{\ssSys}[4]{\left( \begin{array}{c|c}
   #1 & #2 \\
   \hline
   #3 & #4 \\
\end{array}\right)}

\newtheorem{Lemma}{Lemma}[section]
\newtheorem{Assumption}{Assumption}[section]
\newtheorem{Proposition}{Proposition}[section]
\newtheorem{Theorem}{Theorem}[section]

\newtheorem*{Example*}{Example}

\title{Higher-Order Uncoupled Dynamics Do Not Lead to Nash Equilibrium --- Except When They Do\thanks{
Sarah A. Toonsi (stoonsi2@illinois.edu) and Jeff S. Shamma (jshamma@illinois,edu) are with the Department of Industrial and Enterprise Systems Engineering, University of Illinois Urbana-Champaign, Urbana, Illinois, USA.}}
\author{{Sarah A. Toonsi} \and {Jeff S. Shamma}}

\begin{document}

\maketitle

\begin{abstract}
The framework of multi-agent learning explores the dynamics of how individual agent strategies evolve in response to the evolving strategies of other agents. Of particular interest is whether or not agent strategies converge to well known solution concepts such as Nash Equilibrium (NE). Most ``fixed order'' learning dynamics restrict an agent's underlying state to be its own strategy. In ``higher order'' learning, agent dynamics can include auxiliary states that can capture phenomena such as path dependencies. We introduce higher-order gradient play dynamics that resemble projected gradient ascent with auxiliary states. The dynamics are ``payoff based'' in that each agent's dynamics depend on its own evolving payoff. While these payoffs depend on the strategies of other agents in a game setting, agent dynamics do not depend explicitly on the nature of the game or the strategies of other agents. In this sense, dynamics are ``uncoupled'' since an agent's dynamics do not depend explicitly on the utility functions of other agents. We first show that for any game with an isolated completely mixed-strategy NE, there exist higher-order gradient play dynamics that lead (locally) to that NE, both for the specific game and nearby games with perturbed utility functions. Conversely, we show that for any higher-order gradient play dynamics, there exists a game with a unique isolated completely mixed-strategy NE for which the dynamics do not lead to NE. These results build on prior work that showed that uncoupled fixed-order learning cannot lead to NE in certain instances, whereas higher-order variants can. Finally, we consider the mixed-strategy equilibrium associated with coordination games. While higher-order gradient play can converge to such equilibria, we show such dynamics
must be inherently internally unstable. 
\end{abstract}

\section{Introduction}

The framework of learning in games explores how game theoretic solution concepts emerge as the outcome of dynamic processes where agents adapt their strategies in response to the evolving strategies of other agents \cite{fudenberg1998theory,fudenberg2009learning,hart2005adaptive,young2004strategic}.
There is a multitude of specific cases of learning dynamics/game combinations that result in a range of outcomes, including convergence, limit cycles, chaotic behavior, and stochastic stability \cite{hart2000simple, berger2005fictitious, monderer1996fictitiousplay, shamma2004unified, shapley1964sometopics, foster1998nonconvergence, foster1997calibrated, piliouras2014optimization,Young1998Book}.

The emphasis in this literature is on simple adaptive procedures that can result in various solution concepts (e.g., Nash equilibrium, correlated equilibrium, and coarse correlated equilibrium). (A separate concern is the complexity associated with such computations (e.g. \cite{daskalakis2009complexity, babichenko2022communication}).) Using the terminology of the introduction in \cite{hart2013simple}, such dynamics should be ``natural''. 

One ``natural'' restriction for learning is that the dynamics of one agent should not depend explicitly on the utility functions of other agents. Agents adapt their strategies based on considerations derived from their own utility functions. There still will be indirect dependencies (e.g., an agent's payoff depends on the actions of other players which evolve according to the utility functions of other players). 

This restriction was called in \cite{hart2003uncoupled} ``uncoupled'' dynamics. In particular,  \cite{hart2003uncoupled} constructed a specific anti-coordination matrix game for which no uncoupled learning dynamics could converge to the unique mixed-strategy NE. In that setting, the order of the learning dynamics was restricted to match the dimension of the strategy space. Subsequent work \cite{shamma2005dynamic} showed that restriction on the order of the learning dynamics was essential. In particular,  by introducing additional auxiliary states, \cite{shamma2005dynamic} showed that higher-order learning could overcome the obstacle of convergence to NE in the same anti-coordination game while remaining uncoupled.  

Higher-order learning in games can be seen as a parallel to higher-order optimization algorithms, such as momentum-based or optimistic gradient algorithms (e.g., \cite{muehlebach2021optimization, daskalakis2018limit}). Such algorithms utilize the history of gradients to update the underlying search parameter. In this way, there is a path dependency on the trajectory of gradients. Likewise, higher-order learning introduce path dependencies, and hence behaviors, that are not possible by their fixed order counterparts. An early utilization of higher-order learning is in \cite{basar1987relaxation}, in which a player's strategy update uses two stages of history of an opponent's strategies in a zero-sum setting to eliminate oscillations. Similar ideas were used in \cite{conlisk1993adaptation}. Reference \cite{flam2003newtonian} modified gradient based algorithms through the introduction of a cumulative (integral) term. In \cite{shamma2005dynamic}, higher-order dynamics were used to create a myopic forecast of the action of other agents. Reference \cite{laraki2012higher} introduce a version of higher-order replicator dynamics and show that, unlike fixed order replicator dynamics, weakly dominated strategies become extinct. Reference \cite{gao2021passivity} utilizes the system theoretic notion of passivity to analyze a family of higher-order dynamics.

In this paper, we further explore the implications of learning dynamics that are uncoupled. In our setting, an agent's learning dynamics do not depend explicitly on the utility functions of other agents \textit{or even} its own utility function. Rather, learning dynamics depend on the evolution of a payoff vector that is viewed as an externality. We call these ``payoff based'' dynamics. (Such a setting was called ``radically uncoupled'' in \cite{foster2006regret}). When players are engaged in a game, then the payoff stream of one agent depends on the actions of other agents. However, the learning dynamics themselves do not change based on the source of the payoff streams. 

First, we investigate the ability of payoff-based dynamics to converge to mixed-strategy Nash equilibria. We show that for any game with a mixed-strategy NE, there exist payoff-based dynamics that converge locally to that NE. This result is established by making a connection between convergence to NE and the existence of decentralized stabilizing control \cite{wang1973stabilization,davison1990decentralized}. A consequence of the payoff-based structure is that the dynamics also converge to the NE of nearby perturbations of the original game. This outcome is inherited from stability being an open property—i.e., the property of stability is robust to small perturbations.

The form of higher-order learning used for this stability result is higher-order gradient play, which generalize gradient ascent. We show that for any such dynamics, there exists a game with a unique mixed-strategy NE that is unstable under given dynamics. The specific game is a scaled version of the anti-coordination game considered in \cite{hart2003uncoupled}. The tool utilized is a classical analysis method in feedback control systems known as root-locus (e.g., \cite{krall1961extension}, \cite{krall1970rootlocus}), which characterizes the locations of the eigenvalues of a matrix as a function of a scalar parameter. 

A combination of the above results suggests the lack of universality on the side of both learning dynamics and games. While any mixed-strategy NE can be stabilized by suitable higher-order gradient dynamics, any such dynamics can be destabilized by a suitable anti-coordination game.

Finally, we examine the implications of higher-order dynamics being able to converge to the mixed-strategy NE of a $2\times 2$ coordination game, which has two pure NE and one mixed-strategy NE. We show that such higher-order gradient play dynamics must have an inherent internal instability, which makes them unsuitable, if not irrational, as a model of learning.  

The remainder of this paper is organized as follows. Section 2 reviews the notions of finite games and payoff-based learning dynamics. Section 3 introduces both standard gradient play dynamics and higher-order gradient play dynamics. Section 4 provides the mixed-strategy NE stabilization setup and results. Section 5 discusses the non-convergence proof for higher-order gradient play dynamics.  Section 6 discusses mixed-strategy NE that require inherently unstable learning dynamics for stabilization. Section 7 provides simulations and examples. Finally, Section 8 gives concluding remarks. 

\section{Payoff-Based Learning Dynamics}

\subsection{Finite Games}

We consider finite (normal form) games over mixed-strategies. There are $n$ players. The strategy space of player $i\in\theset{1,2,...,n}$ is the probability simplex,
$\Delta(k_i)$, where $k_i$ is a positive integer and $\Delta(\cdot)$ is defined as
$$\Delta(\kappa) = \theset{s\in \mathbb{R}^\kappa \st s_j\ge 0, j = 1, ..., \kappa, \und \sum_{j=1}^\kappa s_j = 1}.$$
 
Define the joint strategy space
$$\mathcal{X} =  \Delta(k_1) \times ... \times \Delta(k_n).$$
The utility function of player $i$ is a function $u_i: \mathcal{X} \rightarrow \mathbb{R}$. 
We sometimes will write 
$$u_i(x_1,...,x_i,...,x_n) = u_i(x_i,x_{-i}),$$
for $x_{i}\in \Delta(k_i)$ and $x_{-i}\in \mathcal{X}_{-i}$, where $\mathcal{X}_{-i} = \Delta(k_1) \times ... \times \Delta(k_{i-1}) \times \Delta(k_{i+1}) \times ... \times \Delta(k_n)$. 

For convenience, we will restrict our discussion to pairwise interactions. That is, the utility function of player $i$ is defined as
\begin{equation}\label{Payoffstructure}
u_i(x_1,...,x_n) = x_i\tr \sum_{j=1\atop j\not= i}^n M_{ij} x_j.
\end{equation} 
for matrices, $M_{ij}$, $j=1,\dots,n$, $j\neq i$.

We can write the utility function of player $i$ as the inner product
$$u_i(x_1,...,x_n) = x_i\tr P_i(x_{-i})$$
where
$$P_i(x_{-i}) = \sum_{j=1\atop j\not= i}^n M_{ij} x_j \in \mathbb{R}^{k_i}.$$
Accordingly, each element of $P_i(x_{-i})$ can be viewed as a payoff that is associated with a component of player $i$'s strategy vector, $x_i$. Note that the dependence of $P_i(\cdot)$ on the $M_{ij}$ is implicit. 

A Nash equilibrium (NE) is a tuple $(x_1^*,...,x_n^*)\in \mathcal{X}$ such that for all $i = 1,..., n$,
$$u_i(x_i^*,x_{-i}^*) \ge u_i(x_i,x_{-i}^*),\quad \forall x_i\in \Delta(k_i).$$
A completely mixed-strategy NE is such that each $x_i^*$ is in the interior of the simplex, i.e., 
$$x^{*}_{i}\in \Int{\Delta(k_{i})}, \quad \forall i,$$ 
where
$$\Int{\Delta(\kappa)} = \theset{s\in \mathbb{R}^\kappa \st s_j> 0, j = 1, ..., \kappa, \und \sum_{j=1}^\kappa s_j = 1}.$$ 

\subsection{Fixed Order Learning}

Our model of learning is a dynamical system that relates trajectories of a payoff vector, $p_i(t)$, to trajectories of the strategy, $x_i(t)$. In particular,
 learning dynamics for player $i$ are specified by a function $f_i:\Delta(k_i)\times \mathbb{R}^{k_i}\rightarrow \mathbb{R}^{k_i}$ according to
$$\dot{x_i}(t) = f_i(x_i(t),p_i(t)),$$
where $x_i(t)\in\Delta(k_i)$ and $p_i:\mathbb{R}_+ \rightarrow  \mathbb{R}^{k_i}$. We assume implicitly that $f_i$ and $p_i$ are such that there exists a unique solution whenever $x_i(0)\in \Delta(k_i)$. We further assume that the dynamics satisfy the invariance property that
\begin{equation}\label{invariance}
x_i(0)\in\Delta(k_i) \Rightarrow x_i(t)\in \Delta(k_i), \quad \forall t\ge 0.
\end{equation}

Note that we define learning dynamics without specifying the source of the payoff vector, $p_i(t)$, hence the terminology ``payoff-based''. This feature will be an important aspect of our formulation in the discussion of higher-order learning and uncoupled dynamics. This separation of the learning dynamics from the source of the payoff stream has been utilized in analyses of population games based on generalizations of contractive (or stable) games associated with the dynamical system property of passivity (e.g., \cite{fox2013population, gao2021passivity, arcak2021dissipativity}).

Only once a player is coupled with other players in a game through their own (possibly heterogeneous) learning dynamics is when we make the connection
\begin{equation}\label{payoff}
p_i(t)=P_i(x_{-i}(t)).
\end{equation} 
This formulation is illustrated in Figure~\ref{fig:uncoupled}, where the $\text{LD}_i$ denote payoff-based learning dynamics that are interconnected through the game matrices, $M_{ij}$.

\begin{figure}
\begin{center}
\includegraphics[width=0.3\textwidth]{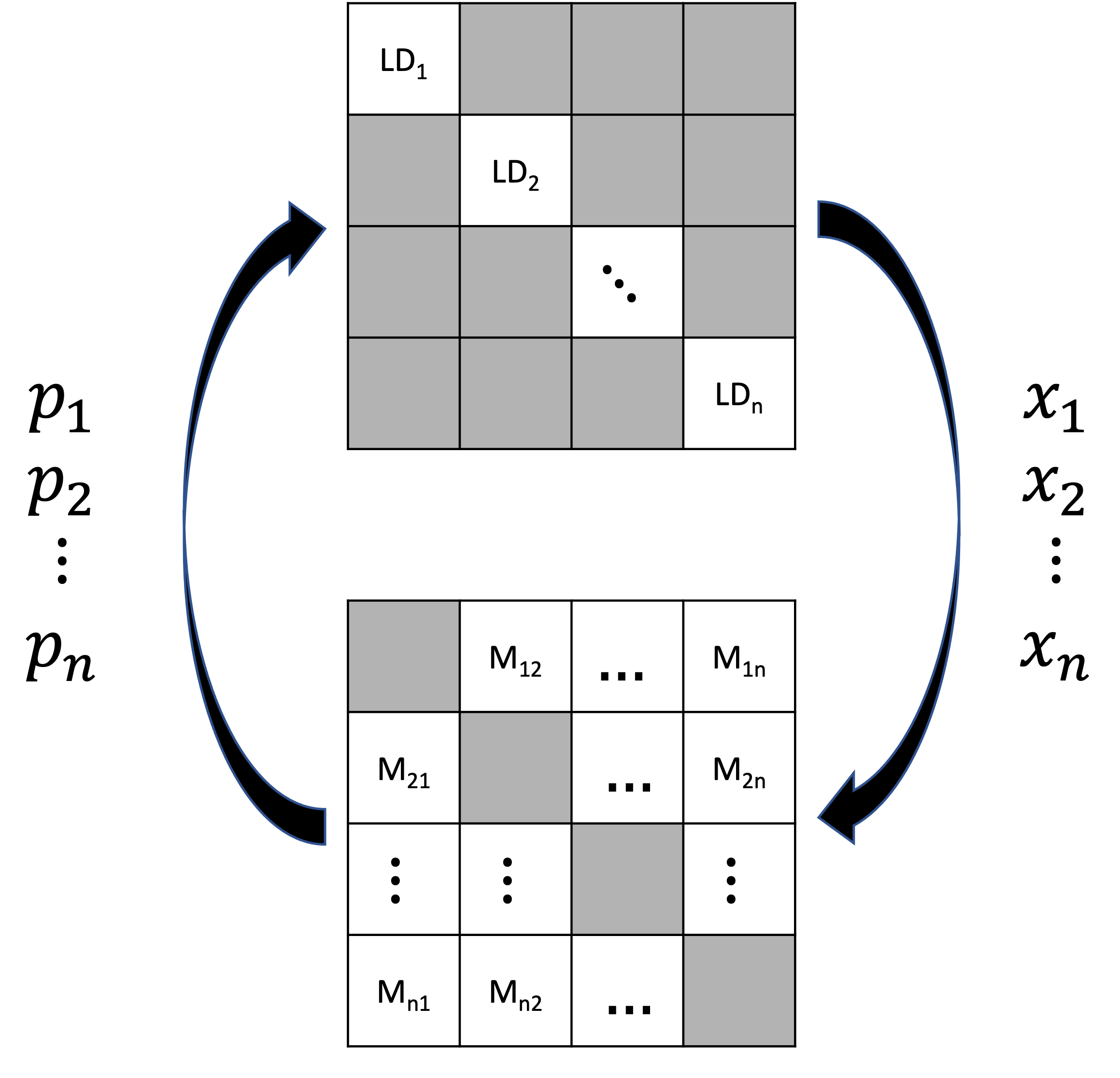}
\caption{Payoff-based learning dynamics ($\text{LD}_i$) in feedback with game matrices ($M_{ij}$).}\label{fig:uncoupled}
\end{center}
\end{figure}

\begin{Example*}[\textbf{Replicator Dynamics}]
The standard representation of replicator dynamics is
$$\dot{x}_i(t) = \diag\left(P_i(x_{-i}(t)) - (x_i(t)\tr P_i(x_{-i}(t)))\cdot \mathbf{1}\right) x_i(t),$$
where $\mathbf{1}$ is a vector of ones of appropriate dimension, and $$\diag\left(v\right):= \Pmatrix{
v_1 & & \\
& \ddots & \\
& & v_k
}, $$ where $v=\Pmatrix{v_1\\\vdots\\v_k}$. The parameters of the matrix game, $M_{ij}$, are built into the function $P_i(\cdot)$. 
In our setting, replicator dynamics take the parallel form
$$\dot{x}_i(t) =\mathbf{diag}\left(p_i(t) -(x_i(t)\tr p_i(t) )\cdot\mathbf{1}\right)x_i(t).$$
There is no specification of game parameters, $M_{ij}$; the payoff stream may or may not be coming from other players; and even if so, other players need not be learning according to replicator dynamics.
\end{Example*}

\begin{Example*}[\textbf{Smooth Fictitious Play}]
First, define the softmax function (or Gibbs distribution)
$$\beta: \mathbb{R}^\kappa \times \mathbb{R}_+ \rightarrow \mathbb{R}^\kappa$$
as
$$\beta(v;T) = \frac{1}{\sum_{j=1}^\kappa e^{v_j/T}} \Pmatrix{e^{v_1/T}\\ \vdots\\ e^{v_\kappa/T}},$$
where $v = \Pmatrix{v_1\\ \vdots\\ v_\kappa}$. Standard smooth fictitious play takes the form
$$\dot{x}_i(t) = \beta(P_i(x_{-i}(t));T) - x_i(t).$$
Again, the parameters of the game are built into the function $P_i(\cdot)$.
In our setting, smooth fictitious play takes the parallel form
$$\dot{x}_i(t) =\beta(p_i(t);T)-x_i(t).$$
As before, the payoff function $p_i(\cdot)$ is not specified.
\end{Example*}

Note that such learning dynamics are uncoupled by construction since each player can only access its own payoff vector. There is no dependence on the payoff stream of other players. Indeed, there is no dependence on the parameters of one's own utility function.


\subsection{Higher-Order Learning}

The learning dynamics described in the previous section have a fixed order associated with the dimension of the strategy space. Higher-order learning dynamics allow for the introduction of auxiliary states as follows. For any fixed order learning dynamics, $f_i$, we can define a higher-order version as
\begin{align*}
\dot{x}_i(t) &= f_i(x_i(t),p_i(t)+\phi_{i}(p_i(t),z_i(t)))\\
\dot{z}_i(t) &= g_i(p_i(t),z_i(t)).
\end{align*}
As before, $x_i(t) \in \Delta(k_i)$ and $p_i: \mathbb{R}_+ \rightarrow \mathbb{R}^{k_i}$. The new variable $z_i\in \mathbb{R}^{\ell_i}$ represents $\ell_i$ dimensional auxiliary states that evolve according to the $p_i$-dependent dynamics, $g_i$. These enter into the original fixed order dynamics through 
$$\phi_i:\mathbb{R}^{k_i} \times \mathbb{R}^{\ell_i}\rightarrow \mathbb{R}^{k_i}.$$
Accordingly,
$$p_i(t) + \phi_{i}(p_i(t),z_i(t))$$
can be viewed as a modified payoff stream that captures path dependencies in $p_i$, and the original learning dynamics react to this modified payoff stream.

As with the fixed-order counterparts, there is no specification of game parameters in higher-order learning dynamics. In order to enforce that the auxiliary states have no effect on the equilibria of games, we make the following assumption.

\begin{Assumption}\label{asm:vanishing}
If $p_i^*$, and $z_i^*$ are an equilibrium of the higher-order dynamics, i.e.,
$$0 = g_i(p_i^*,z_i^*)$$
then
$$\phi_i(p_i^*,z_i^*) = 0.$$
\end{Assumption}

This assumption assures that the auxiliary states represent purely transient phenomena that disappear at equilibrium.
 
\begin{Example*}[\textbf{Linear Higher-Order Dynamics}]
Define
\begin{align*}
\dot{z}_i&=D_i z_i+E_i p_i\\
\phi_{i}&=F_ip_i+G_iz_i,
\end{align*}
where the matrices satisfy
$$F_i - G_i D_i^{-1} E_i = 0.$$
The higher-order dynamics are linear, and the above matrix equality assures that Assumption~\ref{asm:vanishing} is satisfied.
\end{Example*}

\begin{Example*}[\textbf{Anticipatory Higher-Order Dynamics}]
A special case of linear higher-order dynamics is
\begin{align*}
\dot{z}_i &= \lambda (p_i - z_i)\\
\phi_{i} &= \gamma \lambda (p_i - z_i).
\end{align*}
Similar higher-order dynamics were used in \cite{shamma2005dynamic} for smooth fictitious play to overcome the lack of convergence of uncoupled dynamics to NE in the anti-coordination game analyzed by \cite{hart2003uncoupled}. See also \cite{arslan2006anticipatory} for an analysis with replicator dynamics. 
Intuition behind the connection to anticipation can be seen by viewing $\lambda(p_i - z_i)$ as an approximation of $\dot{p}_i$, and so $p_i(t) + \gamma\lambda(p_i(t)-z_i(t))\approx p_i(t+\gamma)$.

Anticipatory higher-order dynamics also can be linked to optimistic optimization algorithms (e.g., \cite{daskalakis2018limit}). An Euler discretization of step size, $h$, results in
$$z_i^+ = z_i + h\lambda (p_i - z_i),$$
with a modified payoff stream of 
$$p_i + \phi_i = p_i + \gamma \lambda (p_i - z_i).$$
Setting $h = 1/\lambda$ and $\gamma = 1/\lambda$ results in
\begin{align*}
p_i + \phi_i &= p_i + (p_i - z_i)\\
&= p_i + (p_i - p_i^-),
\end{align*}
where the superscripts `$+$' and `$-$' indicate the next and previous discrete time steps, respectively.
\end{Example*}

\section{Gradient Play} 

The main results of the paper will examine the behavior of gradient play and higher-order gradient play, which are the focus of this section.

\subsection {Fixed Order Gradient Play}
In gradient play dynamics, a player adjusts its strategy in the direction of the payoff stream, i.e., 
\begin{equation}\label{gradient}
\dot{x}_i = \Pi_\Delta[x_i + p_i] - x_i, 
\end{equation}
where $\Pi_{\Delta}[x]:\mathbb{R}^{n}\rightarrow\Delta\left(n\right)$ is the projection of $x$ into the simplex, i.e.,
$$\Pi_\Delta(x) = \arg\min_{s\in \Delta(n)}\norm{x - s}.$$

The terminology ``gradient play'' stems from the gradient of an agent's utility function in (\ref{Payoffstructure}) with respect to its own strategy, $x_i$, namely
$$\nabla_{x_i} u_i(x_i,x_{-i}) = P_i(x_{-i}) = \sum_{j=1\atop j\not= i}^n M_{ij} x_j.$$ 
As was done in the description of payoff-based learning, we replace $P_i(x_{-i})$ with the payoff stream $p_i$ without regard to the 
game matrices $M_{ij}$. 

Our primary concern will be studying these dynamics near a completely mixed-strategy NE. To this end, let $x^*=(x_1^*,\hdots,x_n^*)$ be an isolated completely mixed-strategy NE. Since $x_i^*\in \Int{\Delta(k_i)}$, the local behavior of the dynamics around $x_i^*$ is characterized by evolution on a lower-dimensional subset. Thus, we can write
\begin{equation}\label{local}
x_i = x_i^* + N_{i}w_i
\end{equation}
where
\begin{equation}\label{NDefinition}
    \mathbf{1}\tr N_{i} = 0\And N_{i}\tr N_{i} = I.
\end{equation}
Therefore, $w_{i}\in\mathbb{R}^{(k_i-1)\times 1}$ represents deviations from $x_i^*$ and satisfies 
$$w_i(t) = N_{i}\tr (x_i(t) - x_i^*).$$

When all players utilize fixed order gradient play, the collective dynamics near a completely mixed-strategy NE take the form 
\begin{equation}\label{StandardGP}
    \dot{w} = \mathcal{M} w,
\end{equation}
where 
\begin{equation}\label{MGrad}
\mathcal{M}=\mathcal{N}\tr \Pmatrix{
0 & M_{12}&M_{13}&\hdots&M_{1n} \\
M_{21}& 0& M_{23} &\hdots & M_{2n}\\
M_{31}& M_{32} & 0 & ... & M_{3n}\\
\vdots & \vdots & \vdots& \ddots&\vdots\\
M_{n1}&M_{n2}&M_{n3} & \hdots&0
}\mathcal{N},
\end{equation}
and
$$
\mathcal{N}=\Pmatrix{
N_{1} & & \\
& \ddots & \\
& & N_{n}}.
$$

Given the zero trace of $\mathcal{M}$, standard gradient play is always unstable at a completely mixed-strategy NE. Also, for this equilibrium to be  isolated, $\mathcal{M}$ must be non-singular. Otherwise, $\mathcal{M}$ has a non-trivial null space leading to an equilibrium subspace.

\subsection{Higher-Order Gradient Play}
\label{sec:Higherordergradient}
We will be interested in a specific form of higher-order gradient play that uses the following structure of higher-order dynamics:
\begin{align*}
\dot{x}_i &= -x_i + \Pi_\Delta\left[x_i + p_i + N_i(G_i \xi_i+ H_i (N_i\tr p_i-v_i))  \right]\\
\dot{\xi}_i & = E_i \xi_i + F_i (N_i\tr p_i-v_i)\\
\dot{v}_i&=N_i\tr p_i-v_i
\end{align*}
for some $E_i$, $F_i$, $G_i$ and $H_i$. Here, the auxiliary states are $z_i = (v_i,\xi_i)$, which enter into the learning dynamics
through 
$$\phi_i(p_i,\xi_i,v_i) = N_i(G_i \xi_i + H_i(N_i\tr p_i - v_i)).$$

\begin{figure}
\begin{center}
\includegraphics[width=0.5\textwidth]{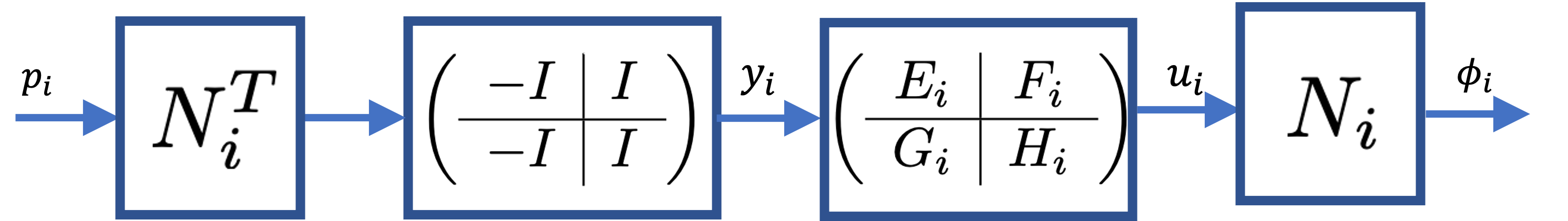}
\caption{Cascade representation of linear higher-order dynamics for gradient play}\label{fig:hogpLinear}
\end{center}
\end{figure}

The motivation behind this structure, illustrated in Figure~\ref{fig:hogpLinear}, assures the enforcement of Assumption~\ref{asm:vanishing}. The payoff stream is first processed by a specific linear system to produce $y_i$ and then by a general linear system to produce $u_i$. The first linear system has the property (see ``washout filters'' in Appendix~\ref{app:washout})
that if $p_i(t)$ converges to a constant, then $y_i(t)$ converges to zero. Accordingly, if the second linear system is stable (determined by the eigenvalues of $E_i$), the output, $u_i$, and hence the modification $\phi_i(t)$ converges to zero. Irrespective of the underlying stability, $\xi_i = 0$ is an equilibrium. Furthermore, $\xi_i$ must be zero at equilibrium whenever $E_i$ is non-singular.

\subsection{Local Stability Analysis}

As before, we can analyze the behavior near a completely mixed-strategy NE $x^* = (x_1^*,...,x_n^*)$. First, define the variables $w_i$ as
$$x_i = x_i^* + N_{i}w_i.$$
Since $x^*$ is a completely mixed NE,
$$p_i^*=\sum_{j\not= i} M_{ij}x_j^* = \alpha_i\mathbf{1}$$
for certain values, $\alpha_i$. The resulting dynamics are
\begin{align*}
\dot{w}_i &= N_i\tr \sum_{j\not= i} M_{ij}(x_j^* + N_j w_j) + G_i \xi_i+ H_i N_i\tr \sum_{j\not= i} M_{ij}(x_j^* + N_j w_j)-H_i v_i\\
\dot{\xi}_i & = E_i \xi_i + F_i N_i\tr \sum_{j\not= i} M_{ij}(x_j^* + N_j w_j)-F_iv_i,\\
\dot{v}_i & = N_i\tr \sum_{j\not= i} M_{ij}(x_j^* + N_j w_j)- v_i,\
\intertext{and using that $N_i\tr p_i^* = 0$ results in}\\[-\baselineskip]
\dot{w}_i &= N_i\tr \sum_{j\not= i} M_{ij}N_j w_j + G_i \xi_i+ H_i \Big( N_i\tr \sum_{j\not= i} M_{ij}N_j w_j - v_i\Big)\\
\dot{\xi}_i & = E_i \xi_i + F_i \Big(N_i\tr \sum_{j\not= i} M_{ij}N_j w_j- v_i\Big)\\
\dot{v}_i & = N_i\tr \sum_{j\not= i} M_{ij} N_j w_j- v_i.
\end{align*}

The collective dynamics near a mixed-strategy NE can be written as
\begin{equation}\label{eq:hogpCL}
\Pmatrix{\dot{w}\\\dot{\xi}\\ \dot{v}}=\Pmatrix{(I+H)\mathcal{M} & G & -H \\F\mathcal{M} & E & -F \\ \mathcal{M} & 0& -I}\Pmatrix{w\\ \xi\\v},
\end{equation}
where the matrices $E$,$F$,$G$, and $H$ are block diagonal with appropriate dimensions and $\mathcal{M}$ is defined in \eqref{MGrad}. 

Local stability of a completely mixed NE is determined by whether the above collective dynamics are stable, i.e., the dynamics matrix in (\ref{eq:hogpCL}) is a stability matrix.
 
\section{Uncoupled Dynamics that Lead to Mixed-Strategy NE}
\label{sec:stabilization}
\subsection{Decentralized Control Formulation}

In the setup of higher-order gradient play, the matrices $(E_i,F_i,G_i,H_i)$ create a dynamical system
that maps (see Figure~\ref{fig:hogpLinear})
$$y_i = N_i\tr p_i - v_i$$
to 
$$u_i = G_i \xi_i + H_i \underbrace{(N_{i}\tr p_i - v_i)}_{y_i}$$
via
$$\dot{\xi}_i = E_i \xi_i + F_i y_i.$$

Let $K_i$ denote the linear dynamical system
$$K_i \sim \ssSys{E_i}{F_i}{G_i}{H_i}.$$
Then the stability of a mixed-equilibrium is tied to the existence of $K_1$, $K_2$, ..., $K_n$, so that the linear system in (\ref{eq:hogpCL}) is  stable.
When the $K_i$ have yet to be determined, we can rewrite (\ref{eq:hogpCL}) as
\begin{subequations}\label{eq:hogpOL}
\begin{align}
\Pmatrix{\dot{w}\cr \dot{v}} &= \Pmatrix{\mathcal{M}& 0\\\mathcal{M}&-I} \Pmatrix{w\cr v} + \Pmatrix{I\cr 0} u,\\
y &= \Pmatrix{\mathcal{M}&-I} \Pmatrix{w\\ v},
\end{align}
\end{subequations}
where 
$$u = \Pmatrix{u_1\\ \vdots\\ u_n}\And y = \Pmatrix{y_1\\ \vdots\\ y_n},$$
and the $y_i$ and $u_i$ are to be related through the $K_i$.

\subsection{Decentralized Stabilization}

Let
$$\mathcal{P} \sim \ssSys{\mathcal{A}}{\mathcal{B}}{\mathcal{C}}{0}$$
with 
\begin{gather*}
\mathcal{A} = \Pmatrix{\mathcal{M}&0\\ \mathcal{M}&-I}, \quad \mathcal{B}= \Pmatrix{I\cr 0}, \quad \mathcal{C}=\Pmatrix{\mathcal{M}&-I}.
\end{gather*}
We first establish that $\mathcal{P}$ can be stabilized by verifying the conditions for stabilizability and detectability (see Appendix~\ref{app:stabilization}). The assumption that $\mathcal{M}$ is non-singular stems from
our interest in isolated NE.

\begin{Proposition}\label{prop:fullrank}
For $\mathcal{M}$ non-singular,  the pair $(\mathcal{A},\mathcal{B})$ is stabilizable,  and the pair $(\mathcal{A},\mathcal{C})$ is detectable. 
\end{Proposition}
\begin{proof}
It is straightforward to verify that for all complex $\lambda$ (not just $\mathbf{Re}[\lambda \ge 0]$),
$$\Pmatrix{\lambda I - \mathcal{A} & \mathcal{B}} = 
\Pmatrix{\lambda I - \mathcal{M}&0&I\\-\mathcal{M}& (\lambda+1)I&0}$$
has full row rank (stabilizability) and
$$\Pmatrix{\lambda I - \mathcal{A} \\ \mathcal{C}} = 
\Pmatrix{\lambda I - \mathcal{M}&0\\-\mathcal{M}& (\lambda+1)I\\\mathcal{M}& -I},$$
has full column rank (detectability). 
\end{proof}

While Proposition~\ref{prop:fullrank} establishes that $\mathcal{P}$ can be stabilized, that property alone is inadequate for our purposes. In particular, in order for the learning dynamics to  be uncoupled, we seek to establish decentralized stabilization according to the partition 
\begin{subequations}\label{eq:decentralized}   
\begin{align}
\Pmatrix{\dot{w}\cr \dot{v}} &= \mathcal{A} \Pmatrix{w\cr v} + \sum_{i=1}^{n} \mathcal{B}_i u_{i}\\
y_{i} &=\mathcal{C}_i  \Pmatrix{w\cr v},
\end{align}
\end{subequations}
where
\begin{subequations}\label{eq:mathcalABC}
\begin{gather}
\mathcal{A} = \Pmatrix{\mathcal{M}&0\\ \mathcal{M}&-I},\quad \mathcal{B}_i =\Pmatrix{\mathcal{E}_i \\ 0}\\
\mathcal{C}_i = \Pmatrix{\mathcal{M}_{i\bullet}& -\mathcal{E}_i\tr}.
\end{gather}
\end{subequations}
Here, $\mathcal{M}_{i\bullet}$ denotes the $i^\text{th}$ block row of $\mathcal{M}$, i.e.,
$$\mathcal{M}_i = \Pmatrix{
N_i\tr M_{i1} N_1 &
...&
N_i\tr M_{i(i-1)} N_{i-1}&
0&
N_i\tr M_{i(i+1)} N_{i+1}&
...&
N_i\tr M_{in}N_n}
$$
and $$\mathcal{E}_i\tr = \Big(0\quad ... \quad 0 \quad \underbrace{I}_{i^\text{th} \text{\ position}}\quad  0 \quad ... \quad 0\Big),$$
where $I$ has dimension (suppressed in the notation) of $k_i - 1$.

\begin{Theorem}\label{thm:stabilize} For any isolated (i.e., $\mathcal{M}$ is non-singular) completely mixed-strategy NE, there exist uncoupled higher-order gradient play dynamics such that (\ref{eq:hogpCL}) is stable.
\end{Theorem}

\noindent The proof of Theorem~\ref{thm:stabilize} relies on the conditions of Theorem~\ref{thm:davison} and is presented in Appendix~\ref{app:proof}.

Theorem~\ref{thm:stabilize} should be viewed as a statement regarding whether uncoupled learning in itself is a barrier to learning dynamics leading to NE. The theorem makes no claim that
the higher-order learning dynamics are interpretable (e.g., as in anticipatory learning). Nor does the theorem offer guidance on how agents may construct 
the matrices of higher-order learning that lead to convergence. In the next section, we will see that, while the structure is universal, any specific set of parameters is not universal in that one can construct a game for which they do not lead to NE. 

Despite the lack of universality, there is an inherent robustness that is a consequence of stability. The following follows from standard arguments on linear systems.

\begin{Proposition}\label{prop:near} Let the $(E_i,F_i,G_i,H_i)$ and $M_{ij}$, $i=1,...,n$ and $j=1,...,n$, be such that (\ref{eq:hogpCL}) is stable. Then there exists a  $\delta>0$ such that (\ref{eq:hogpCL}) is stable with the $M_{ij}$ replaced by any $\tilde{M}_{ij}$ as long as $\norm{\tilde{M}_{ij} - M_{ij}} < \delta$ for all $i=1,...,n$ and $j=1,...,n$.
\end{Proposition}

In words, this proposition guarantees that learning dynamics that lead to NE for a specific game continue to do so for nearby games.

\subsection{Stabilization through a Single Higher-Order Player}

The previous section's analysis allowed all players to utilize higher-order learning. In some cases, it may not be necessary that all players utilize higher-order learning. In this section, we present sufficient conditions under which a single player using higher-order gradient play with the remainder utilizing fixed order gradient play can still lead to NE.

 \begin{Assumption}\label{asm:full} \ \\[-\baselineskip]
 \begin{itemize}
 \item[A.] Let $(w,\lambda)$ be a left eigenvalue pair of $\mathcal{M}$, i.e.,
$$w\tr \mathcal{M} = \lambda w\tr,$$
with $\mathbf{Re}[\lambda] \geq 0$ and
$$w\tr = \Pmatrix{w_1\tr &  w_2\tr & ... & w_n\tr}$$
partitioned consistently with (\ref{MGrad}). 
Then $w_i\not= 0$ for all $i$.

\item[B.] Let $(v,\lambda)$ be a right eigenvalue pair of $\mathcal{M}$, i.e., 
$$\mathcal{M}v = \lambda v,$$
with $\mathbf{Re}[\lambda] \geq 0$ and
$$v = \Pmatrix{v_1\\ v_2\\ \vdots\\ v_n}$$
partitioned consistently with (\ref{MGrad}). Then $v_i\not= 0$ for all $i$.
\end{itemize}
 \end{Assumption}
 
Recall the definitions of $\mathcal{A}$, $\mathcal{B}_i$, and $\mathcal{C}_i$ from (\ref{eq:mathcalABC}).

\begin{Proposition}\label{prop:single} Let $\mathcal{M}$ be non-singular and satisfy Assumption~\ref{asm:full}. Then for any $i$, the pair $(\mathcal{A},\mathcal{B}_i)$ is stabilizable and
the pair $(\mathcal{A},\mathcal{C}_i)$ is detectable.
\end{Proposition}

\begin{proof}
For stabilizability, we need to examine the row rank of
$$\Pmatrix{\mathcal{A} - \lambda I&\mathcal{B}_i}$$
Suppose there is row rank deficiency, i.e., 
\begin{align*}
\Pmatrix{w\tr & z\tr} \Pmatrix{\lambda I - \mathcal{A}&\mathcal{B}_i }&= \Pmatrix{w\tr & z\tr} \Pmatrix{\lambda I - \mathcal{M}&0&\mathcal{E}_i\\ -\mathcal{M}&(\lambda+1) I&0}\\
&= \Pmatrix{0&0&0}
\end{align*}
for some non-zero $\Pmatrix{w&z}$. Then necessarily, $z = 0$ and $(w,\lambda)$ form a left eigenvector pair. By assumption, $w\tr \mathcal{E}_i \not= 0$, which is a contradiction.

Likewise, for detectability, we need to examine the column rank of 
$$\Pmatrix{\mathcal{C}_i\\ \mathcal{A} - \lambda I}.$$
Suppose there is column rank deficiency, i.e.,
$$\Pmatrix{\lambda I - \mathcal{M} &0\\-\mathcal{M} & (\lambda+1)I\\ \mathcal{M}_{i\bullet} & -\mathcal{E}_i\tr}\Pmatrix{v\\ z} = 0$$
for some non-zero $\Pmatrix{v\\ z}$.
Then 
$$\lambda v = \mathcal{M} v\And (\lambda + 1) z = \mathcal{M}v,$$
together imply that
$$z = \frac{\lambda}{\lambda + 1} v.$$
Furthermore,
$$z_i = \mathcal{M}_{i\bullet} v = \lambda v_i.$$
Putting these together, it must be that $\lambda v_i = 0$. By assumption, $\lambda\not=0$ (non-singularity of $\mathcal{M}$), and so $v_i\not=0$ (Assumption~\ref{asm:full}), resulting in a contradiction.
\end{proof}

As a consequence of Proposition~\ref{prop:single}, it is possible for a completely mixed-strategy NE to be stabilized where a single player utilizes higher-order gradient play with the remaining players utilizing fixed order gradient play.

\section{Non-Convergence to NE in Higher-Order Gradient Play}
 
In this section, we show that linear higher-order gradient play dynamics need not lead to NE. Given any such dynamics, we will construct an anti-coordination game with a unique mixed-strategy NE that is unstable under given dynamics. 

\subsection{The Jordan Anti-coordination Game}
The Jordan anti-coordination game was introduced in \cite{jordan1993three} and was later used in \cite{hart2003uncoupled} to prove that fixed-order uncoupled learning dynamics do not lead to NE. The game consists of three players with the following utility functions
\begin{align*}
u_1(x_1,x_2) &= x_1\tr M_{12} x_2\\
u_2(x_2,x_3) &= x_2\tr M_{23} x_3\\
u_3(x_3,x_1) &= x_3\tr M_{31} x_1,
\end{align*}
where
$$M_{12} =M_{23} = M_{31} = \Pmatrix{0&1\\1&0}.$$
The game has a unique mixed-strategy NE at
$$x_1^*= x_2^*= x_3^* = \Pmatrix{\frac{1}{2}\\[5pt] \frac{1}{2}}.$$
We will let $\Gamma(\mu)$ denote the Jordan anti-coordination game but with the utility function of player 1 modified to
$$u_1(x_1,x_2) = x_1\tr (\mu M_{12}) x_2,$$
where $\mu \in \mathbb{R}_+$.  Since scaling payoffs does not change the nature of the game, $\Gamma(\mu)$ has the same unique NE as $\Gamma(1)$. 

\subsection{Destabilization Using Rescaled Anti-Coordination}
 

When all three players use variants of linear higher-order gradient play in the Jordan anti-coordination game, we get the following dynamics  
\begin{align*}
\dot{x}_1 &= -x_1 + \Pi_\Delta\left[x_1 + \mu M_{12}x_2 + N\big(G_1 \xi_1 + h_1 (\mu N\tr  M_{12}x_2 - v_1)\big)\right]\\
\dot{\xi}_1 & = E_1 \xi_1 + F_1 (\mu N\tr  M_{12}x_2 - v_1)\\
\dot{v}_1 &= \mu N\tr  M_{12}x_2 - v_1,\\
\dot{x}_2 &= -x_2 + \Pi_\Delta\left[x_2 + M_{23}x_3 + N\big(G_2 \xi_2 + h_2 (N\tr M_{23}x_3 - v_2)\big)\right]\\
\dot{\xi}_2 & = E_2 \xi_2 + F_2 (N\tr M_{23}x_3 - v_2)\\
\dot{v}_2 &= N\tr M_{23}x_3 - v_2,\\
\dot{x}_3 &= -x_3 + \Pi_\Delta\left[x_3 + M_{31}x_1 + N\big(G_3\xi_3 + h_3 (N\tr M_{31}x_1 - v_3)\big)\right]\\
\dot{\xi}_3 & = E_3 \xi_3 + F_3 (N\tr  M_{31}x_1 - v_3)\\
\dot{v}_3 &= N\tr  M_{31}x_1 - v_3.
\end{align*}
To study the local behavior of the dynamics around the unique mixed-strategy NE, we define 
$$w_1(t)=N\tr(x_1(t)-x_1^*),\quad w_2(t)=N\tr(x_2(t)-x_2^*), \quad w_3(t)=N\tr(x_3(t)-x_3^*),$$
where 
$$N=\Pmatrix{\frac{1}{\sqrt{2}}\\[5pt] \frac{-1}{\sqrt{2}}}.$$

We can analyze the local stability of the mixed-strategy NE through the local collective dynamics

%
%
%
%
%
%
%

$$\Pmatrix{\dot{w}\\\dot{z}}= J^r \Pmatrix{w\\ z},$$
where 
$$J^r:=\Pmatrix{
0&-\mu(1+h_1)&0& G_1 &-h_1&0& 0&0&0\\

0&0&-(1+h_2)&
0&0&G_2&-h_2&0&0\\

-(1+h_3)&0&0&
0&0&0&0&G_3&-h_3\\

0&-\mu F_1&0& E_1 &-F_1&0& 0&0&0\\
0&-\mu &0& 0 &-1&0& 0&0&0\\

0&0&-F_2&
0&0&E_2&-F_2&0&0\\
0&0&-1&
0&0&0&-1&0&0\\
-F_3&0&0&
0&0&0&0&E_3&-F_3\\
-1&0&0&
0&0&0&0&0&-1
}.$$ 

In the above we have $w=(w_1,w_2,w_3)$, $z=(\xi_1,v_1,\xi_2,v_2,\xi_3,v_3)$, and we substituted $N\tr M_{12}N=N\tr M_{23}N=N\tr M_{31}N=-1$. Our main result suggests that even if the dynamics could stabilize the mixed-strategy NE of the nominal Jordan anti-coordination game, they would not be able to do it for some rescaled version of the Jordan anti-coordination game. 

\begin{Proposition}\label{prop:nonconvergance}
If linear higher-order gradient play dynamics are locally exponentially stable at the unique NE of $\Gamma(1)$, then there exists $\mu > 0$ such that the unique NE of   $\Gamma(\mu)$ is unstable under such dynamics.
\end{Proposition}
The following two subsections will provide two proofs for Proposition~\ref{prop:nonconvergance}. In both proofs, we use root-locus arguments. 
\subsubsection{A proof using sufficiently large $\mu$}
We will exploit the structure of $J^r$ to study the behavior of its eigenvalues as $\mu$ varies. Consider first the following lemma. 
\begin{Lemma}\label{lemma:5.1}
Let $A \in \mathbb{R}^{n \times n}$, $B \in \mathbb{R}^{n \times 1} \text { and } C \in \mathbb{R}^{1 \times n}$. If $CB =CAB = 0$, and $CA^mB\not=0$ for some $m \ge 2$. Then for sufficiently large $\mu > 0$, $A-\mu BC$ is not a stability matrix.
\end{Lemma}

\begin{proof}
Define
$$H(s) = C(sI-A)^{-1}B.$$
Since $H(s)$ is a rational function, we can write it as
$$H(s) = \frac{p(s)}{q(s)},$$
for polynomials $p$ and $q$ that have no common roots. The assumption that $CA^mB \not= 0$ for some $m$ assures that $H(s)$ is not identically equal to zero. 

Suppose that for some $\mu$ and $s'$ that is not an eigenvalue of $A$,
$$q(s') + \mu p(s') = 0.$$
Then $s'$ is an eigenvalue of $A - \mu BC$, since
\begin{align*}
\Det{s' I - (A - \mu BC)} &= \Det{s' I - A}\Det{I + \mu(s'I - A)^{-1}BC}\\
&= \Det{s'I - A}(1 + \mu C(s'I - A)^{-1}B)\\
&=\Det{s'I-A}(q(s') + \mu p(s')) \frac{1}{q(s')}.
\end{align*}
Note that the roots of $q(s)$ are a subset of the roots of $\Det{sI - A}$.
For sufficiently large $\magn{s}$, we can rewrite $H(s)$ as
\begin{align*}
H(s) &= \frac{1}{s} C(I - \frac{1}{s}A)^{-1} B\\
&=\frac{1}{s}C \Big( \sum_{k = 0}^{\infty} \frac{1}{s^k} A^k \Big) B.
\end{align*}
By assumption, $CB = 0$ and $CAB = 0$, which implies that the first two terms of the series equal zero. Accordingly,
$$\lim\sup_{\magn{s}\rightarrow \infty} \magn{s}^3 {\magn{H(s)}} < \infty.$$
The main implication here is that the degree of $q(s)$ is at least 3 more than the degree of $p(s)$. Root-locus arguments in \cite{krall1961extension} and \cite{krall1970rootlocus} (asymptote rule) imply that
$$q(s) + \mu p(s)$$ 
has roots with positive real parts for large $\mu$.
\end{proof}

 The matrix $J^r$ can be written in the form $A - \mu BC$ with

%
%
%
%
%
%
%
%
%
%

\begin{subequations}
\begin{gather}\label{Jrshape}
\small
A = 
\Pmatrix{
0&0&0& G_1 &-h_1&0& 0&0&0\\
-(1+h_3)&0&0&
0&0&G_3&-h_3&0&0\\
0&-(1+h_2)&0&
0&0&0&0&G_2&-h_2\\
0&0&0& E_1 &-F_1&0& 0&0&0\\
0&0 &0& 0 &-1&0& 0&0&0\\
-F_3&0&0&
0&0&E_3&-F_3&0&0\\
-1&0&0&
0&0&0&-1&0&0\\
0&-F_2&0&
0&0&0&0&E_2&-F_2\\
0&-1&0&
0&0&0&0&0&-1
}
\quad B = \Pmatrix{ h_1+1 \\ 0 \\ 0 \\  F_1\\ 1 \\ 0\\0\\0\\0}\\
C = \Pmatrix{0&0&1&0&0&0&0&0&0},
\end{gather}
\end{subequations}
where we reordered the variables according to 
$(w_1, w_3, w_2, \xi_1,v_1,\xi_3,v_3,\xi_2,v_2)$ for convenience.
Following the same arguments in proving Lemma~\ref{lemma:5.1}, if $CA^mB= 0$ for all $m$, then eigenvalues of $A-BC$ are the same as the eigenvalues of $A$. 
Writing $A$ in block matrix form yields
$$A = \Pmatrix{A_{11}&A_{12}\\
A_{21}&A_{22}},$$
where $A_{11}$ is $3\times 3$. Notice that $A_{11}$ is strictly lower triangular, and $A_{21}$ is strictly block lower triangular. Now examine
$$\Det{sI - A} = \Det{sI - A_{11}}\Det{(sI - A_{22}) - A_{21}(sI - A_{11})^{-1}A_{12}}.$$
One can show that $A_{21}(sI - A_{11})^{-1}A_{12}$ is strictly block lower triangular. Therefore
$$\Det{sI - A} = \Det{sI - A_{11}}\Det{sI - A_{22}}.$$
Thus, $A$ has eigenvalues at 0 with multiplicity 3 or more because of $A_{11}$. By exponential stability, there exists $m\geq2$ such that $CA^mB\neq 0$. Since $CB=0$, and $CAB=0$, we can now apply Lemma~\ref{lemma:5.1} to show that $\exists$ $\mu$ such that $A-\mu BC$ is not a stability matrix.

\subsubsection{A proof using sufficiently small $\mu$}
The scaling $\mu$ was large in the previous proof, and the analysis was asymptotic. Non-convergence over bounded games, e.g., 
$$ \lVert M_{ij} \rVert_\alpha < 1 \quad \forall i,j$$
is considered next.

\begin{Lemma}\label{lemma:5.2}
Let $A \in \mathbb{R}^{n \times n}$, $B \in \mathbb{R}^{n \times 1}$$ \text { and } C \in \mathbb{R}^{1 \times n}$. Assume that $A$ has eigenvalues at $0$ with multiplicity 3 or more. Then for sufficiently small $\mu > 0$, $A-\mu BC$ is not a stability matrix.
\end{Lemma}
\begin{proof}
    As in the proof of Lemma~\ref{lemma:5.1}, we have
$$\Det{s I - (A - \mu BC)} = \Det{sI-A}(q(s) + \mu p(s)) \frac{1}{q(s)}.$$
Recall that the roots of $q(s)$ are a subset of the roots of $\Det{sI - A}$. If $q(s)$ does not have at least 3 roots at zero, then $A-\mu BC$ is not a stability matrix. Otherwise, root-locus arguments in \cite{krall1961extension}, and \cite{krall1970rootlocus} (angle of departure rule)  imply there exist roots of $q(s) + \mu p(s)$ with positive real parts for small $\mu$.
\end{proof} 

Using the structure of $J^r$ in \eqref{Jrshape} and the fact that $A$ has eigenvalues at 0 with multiplicity 3 or more, one can directly use Lemma~\ref{lemma:5.2} to show the existence of a sufficiently small offending $\mu$.
 
\subsection{Discussion}\label{sec:muDiscussion}

In both proofs, we did not construct a single game for which all higher-order gradient play  dynamics do not lead to NE. Instead, for any such dynamics,  we show a ``challenger'' $\Gamma(\mu)$ is destabilizing. It is also worth mentioning that the choice of the Jordan anti-coordination game is not limiting. The implications of the root-locus arguments will hold given any game with the appropriate structure and number of players.

The results might be puzzling because, for all $\mu>0$, all games $\Gamma(\mu)$ are strategically equivalent. Convergence guarantees for learning dynamics are usually established amongst classes of games. Thus, it is generally expected that dynamics will behave similarly for all games in a particular class. In this case, we design linear learning dynamics that are affected by simple rescaling of the payoff matrices.

\section{Strong Stabilization of Mixed-Strategy NE}
Results from Section~\ref{sec:stabilization} imply that the mixed-strategy NE in a two-player $2\times2$ (identical-interest) coordination game can be stabilized. Here, we argue why dynamics that stabilize this mixed-strategy equilibrium are not reasonable. Specifically, we show that such dynamics \textit{must be} inherently unstable as an open system, i.e., as dynamics that respond to an exogenous payoff stream, and this instability is problematic with respect to such payoffs. 

First, we inspect which type of mixed-strategy NE requires unstable learning dynamics for stabilization. For this purpose, consider the system in \eqref{eq:decentralized} for $n=k_1=k_2=2$:
\begin{equation}\label{twoplayers}
 \mathcal{A} = \Pmatrix{0&m_{12} &0&0\\ m_{21}&0&0&0\\ 0&m_{12}&-1&0\\ m_{21}&0&0&-1}\quad \mathcal{B} = \Pmatrix{1&0\\0&1\\0&0\\0&0}
\end{equation}
$$\mathcal{C} = \Pmatrix{0&m_{12}&-1&0\\ m_{21}&0&0&-1}.$$
Around an isolated mixed-strategy NE, the matrix $\mathcal{A}$ above should be non-singular. Accordingly, it must be that $m_{12}\neq0$ and $m_{21}\neq0$. The ability to stabilize a system via another stable system is referred to as \emph{strong stabilization} (see Appendix~\ref{app:stabilization}).
The next proposition gives a sufficient condition under which an isolated mixed-strategy NE is not strongly stabilizable.

\begin{Proposition}\label{prop:strong}
 If $m_{12}m_{21} > 0$, then \eqref{twoplayers} is not strongly stabilizable.
\end{Proposition}
\begin{proof} Reference \cite{youla1974singleloop} presents a necessary and sufficient condition for strong stabilizability. First, we compute 
\begin{align*}
T(s) &= \mathcal{C}(sI-\mathcal{A})^{-1}\mathcal{B}\\
&= \frac{s}{s+1} \mathcal{M}(sI - \mathcal{M})^{-1},
\end{align*}
where we used
$$(sI - \mathcal{A})^{-1} = \Pmatrix{ (sI-\mathcal{M})^{-1}&0\\ \frac{1}{s+1} \mathcal{M}(sI - \mathcal{M})^{-1}& \frac{1}{s+1} I}.$$
There is blocking zero (i.e., $T(s) = 0$) at $s = 0$ and $\magn{s} \rightarrow \infty$. 
According to \cite{youla1974singleloop}, a necessary and sufficient condition for strong stabilizability is that there should be an \textit{even} number of eigenvalues in between such pairs of real zeros. This property is known as the ``parity interlacing principle''.
The eigenvalues of $\mathcal{A}$ are
$$-1, -1, \pm \sqrt{m_{12}m_{21}},$$
and so there is a single (and hence, an odd number) real eigenvalue in between two
real zeros of $T(s)$.  Accordingly, \eqref{twoplayers} is not strongly stabilizable.
\end{proof}

The nature of the game can be inferred from the scalars $m_{12}$ and $m_{21}$. For example in zero-sum games we have $M_{12} = -M_{21}\tr$, which gives
$$m_{12} = N\tr M_{12} N = N\tr M_{12}\tr N = -N\tr M_{21} N = -m_{21}.$$
In coordination games, we have $M_{12} = M_{21}$, which gives
$$m_{12} = N\tr M_{12}N = N\tr M_{21} N = m_{21}.$$
Therefore, the mixed-strategy NE in a coordination game is not strongly stabilizable. 

Now let us examine the implications of inherently unstable learning dynamics. A reasonable expectation of learning dynamics is that in the case of a constant payoff vector, i.e., $p_i(t) \equiv p^*$, then we expect
\[ \lim_{t\rightarrow \infty} x_{i}(t) = \beta(p^*) ,\]
where $\beta(p^*) $ is a best response, i.e.,
\[\beta(p^*)=\argmax_{x_i\in\Delta(k_i)} x_i\tr p^*. \]
For any higher-order gradient play dynamics, if $E_i$ is a stability matrix, then whenever $p_i(t)\equiv p^*$ for some constant vector $p^*$, one can show $\xi_i(t)\rightarrow0$, which implies that  $x_i(t)$ is generated by standard gradient play dynamics in the limit. However, if $p_i(t)\equiv p^*$ and the dynamics are inherently unstable, the term $N_iG_i \xi_i(t)$ need not vanish. Indeed, one can construct $p^*$ such that $x_i(t)$ does not converge to the best response of $p^*$ (see the example in Section~\ref{sec:coordinationexm}). The inability of learning dynamics to converge to the best response of a constant payoff  vector does not reflect ``natural'' behavior.

\section{Simulations and Examples}

\subsection{Jordan Anti-Coordination Game: Stabilization Through a Single Player}

We will now attempt to stabilize the Jordan anti-coordination game's mixed-strategy NE, allowing only one player to use higher-order learning while others continue to use standard gradient play. To do that, we must check Assumption~\ref{asm:full} for the Jordan anti-coordination game. Finding the right/left eigenvalues and eigenvectors of the matrix 
$$\mathcal{M}=\Pmatrix{0 &-1 &0\\0& 0 &-1\\-1& 0 &0},$$
we see that Assumption~\ref{asm:full} is satisfied. Therefore, we let $\xi_1\in\mathbb{R}$ and choose $H_1=\gamma\lambda$, $G_1=-\gamma\lambda$, $F_1=\lambda$ and $E_1=-\lambda$, where $\lambda=50$ and $\gamma=5$. Such dynamics resemble anticipatory gradient play but on the filtered low-dimensional payoff. Figure~\ref{fig:singleJordan} illustrates convergence to NE using these parameters in the Jordan anti-coordination game.
\begin{figure}
    \centering
    \includegraphics[width=.5\textwidth]{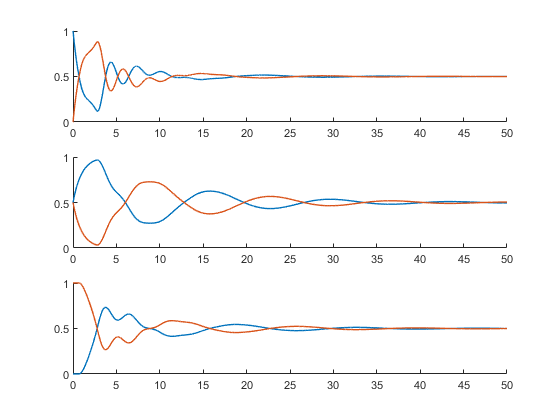}
    \caption{Player strategies in Jordan anti-coordination game: Single-player stabilization.}
    \label{fig:singleJordan}
\end{figure}

\subsection{Robust Stabilization of the Jordan Anti-Coordination Game}
Let us now demonstrate the robustness of the learning dynamics designed in the previous example. According to Proposition~\ref{prop:near}, such dynamics should stabilize nearby games. Consider for example the following perturbed game
\begin{align*}
    M_{12}= \Pmatrix{0&1\\1&0}+\tilde{M}_1 , \quad  M_{23}= \Pmatrix{0&1\\1&0}+\tilde{M}_2, \quad  
    M_{31}= \Pmatrix{0&1\\1&0}+\tilde{M}_3.
\end{align*} The entries of each $\tilde{M}_i$ are sampled from a zero mean normal distribution,  and we only consider small standard deviations to avoid nearby games with pure NE. Figure~\ref{fig:randomJordan} shows how the dynamics from Section 7.1 stabilize the unique mixed-strategy NE of two perturbed Jordan anti-coordination games with standard deviations of 0.3 and 0.5.
\begin{figure}
     \centering
     \begin{subfigure}[b]{0.4\textwidth}
         \centering
         \includegraphics[width=\textwidth]{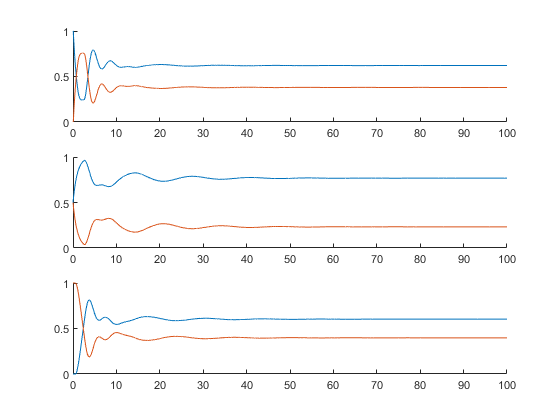}
         \caption{Standard deviation 0.3.}
         \label{Sd 0.3.}
     \end{subfigure}
     \hfil
     \begin{subfigure}[b]{0.4\textwidth}
         \centering
         \includegraphics[width=\textwidth]{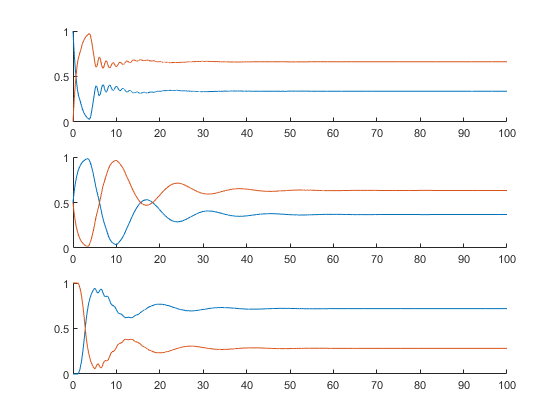}
         \caption{Standard deviation 0.5.}
         \label{Sd 0.3.}
     \end{subfigure}
    \caption{Player strategies in randomly perturbed Jordan anti-coordination game}
        \label{fig:randomJordan}
\end{figure}

Now consider perturbations of the form 
\begin{align*}
    M_{12}= \Pmatrix{\delta_1&1\\1&\delta_1},
     \quad M_{23}= \Pmatrix{\delta_2&1\\1&\delta_2}, \quad 
    M_{31}= \Pmatrix{\delta_3&1\\1&\delta_3},
\end{align*}
where $0<\delta_i<1$. For small values of $\delta_i$, the dynamics from the previous example converge to the unique mixed-strategy equilibrium at $$x_1^*= x_2^*= x_3^* = \Pmatrix{\frac{1}{2}\\[5pt] \frac{1}{2}}.$$ 
 Convergence for values $\delta_1=0.3877$, $\delta_2=0.1446$, and $\delta_3=0.1352$ is illustrated in Figure~\ref{fig:Smalldelta}. In Figure~\ref{fig:largerdelta}, we show how the dynamics no more stabilize the unique mixed-strategy NE for larger perturbations of values $\delta_1=0.8831$, $\delta_2=0.4259$ and $\delta_3=0.7546 $.

\begin{figure}
     \centering
     \begin{subfigure}[b]{0.4\textwidth}
         \centering
         \includegraphics[width=\textwidth]{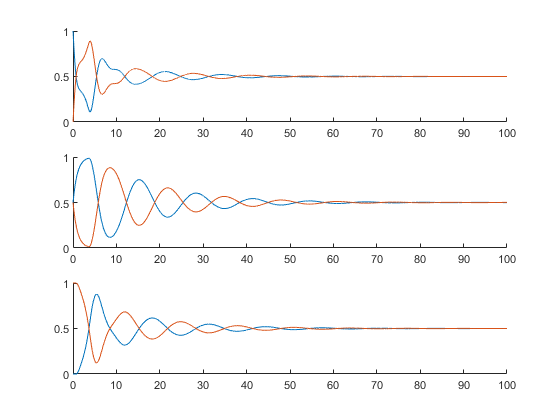}
         \caption{Small $\delta_i$'s.}
         \label{fig:Smalldelta}
     \end{subfigure}
     \hfil
     \begin{subfigure}[b]{0.4\textwidth}
         \centering
         \includegraphics[width=\textwidth]{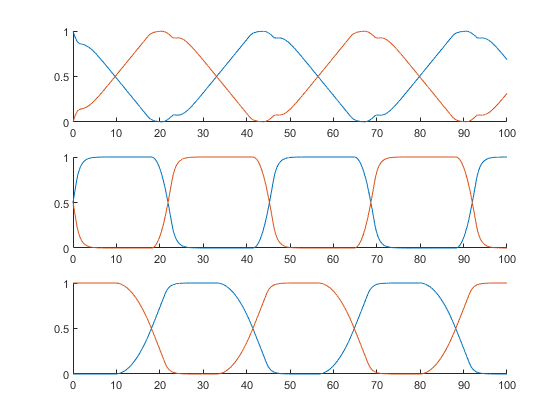}
         \caption{Larger $\delta_i$'s.}
         \label{fig:largerdelta}
     \end{subfigure}
    \caption{Player strategies in diagonally perturbed Jordan anti-coordination game}
        \label{fig:PerturbedJordan}
\end{figure}
\subsection{Rescaled Jordan Anti-Coordination Game}
We now present an example of how different scalings of a payoff matrix affect the stability of  linear higher-order gradient dynamics in the Jordan anti-coordination game. We will let all three players use higher-order gradient play and study the dynamics for different $\Gamma(\mu)$ games. We let $\xi_i\in\mathbb{R}$ for all $i$ and choose the following parameters for all players: $h_i=\gamma\lambda$, $G_i=-\gamma_2\lambda$, $F_i=\lambda$ and $E_i=-\lambda$. Let $\gamma_1=1$, $\gamma_2=.8$, and $\lambda=5$. When $\mu=1$, the matrix $J^r$ is a stability matrix, and we get convergence to NE.  

\begin{figure}
    \centering
    \includegraphics[width=.5\textwidth]{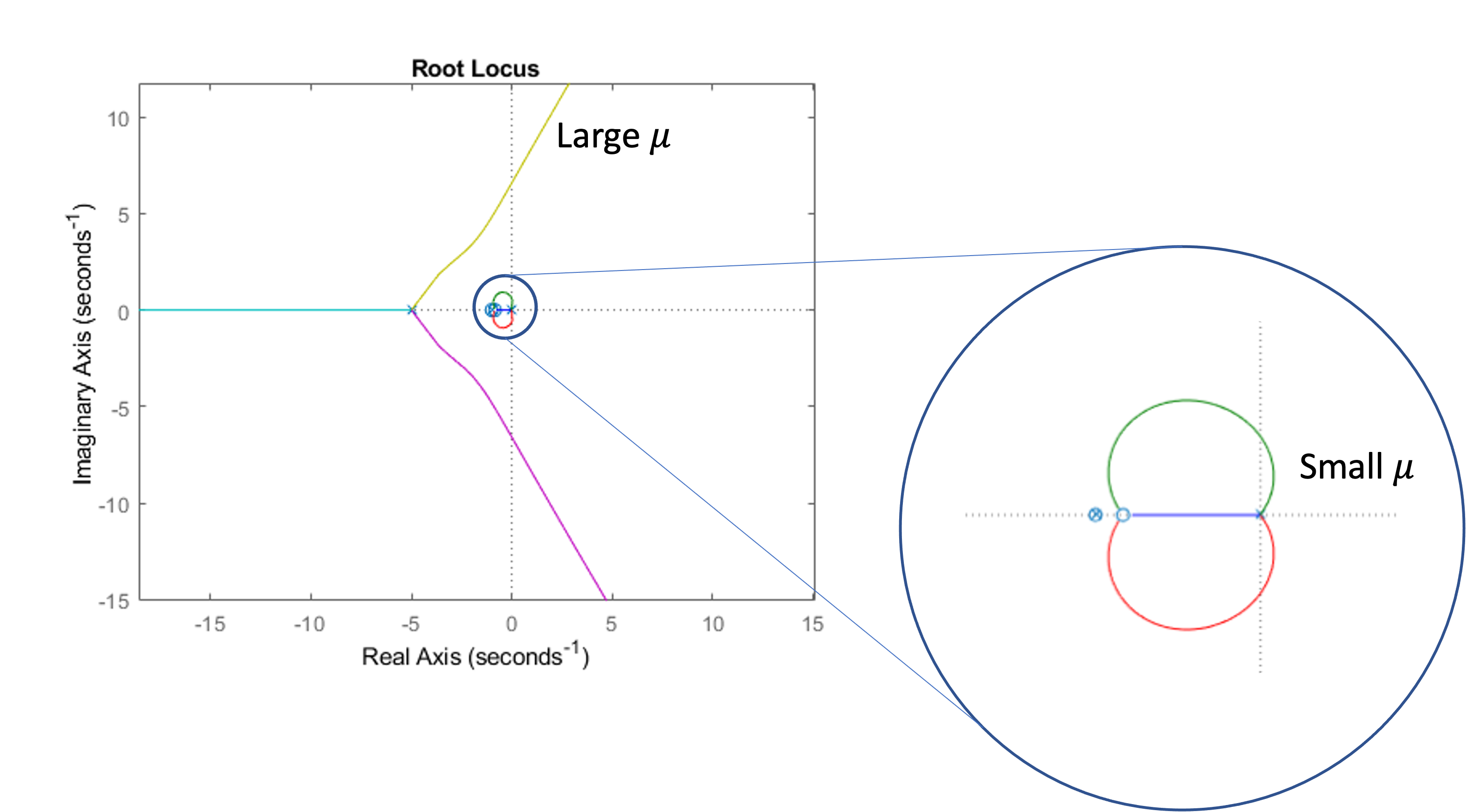}
    \caption{Root-locus plot of eigenvalue locations for $\mu\in (0,\infty)$.}
    \label{fig:rootlocusCombo}
\end{figure}

The root-locus plot in Figure~\ref{fig:rootlocusCombo} shows the eigenvalues of $J^r$ as $\mu$ varies from $0$ to $\infty$. As seen in the plot, there are eigenvalues with 
positive real parts for large $\mu$ (approximately $\mu > 3$). Figure~\ref{fig:largemu} illustrates instability for $\mu=5$. The root-locus plot in Figure~\ref{fig:rootlocusCombo} also shows that the eigenvalues of $J^r$ that start at the origin first drift into the right-half-plane for small $\mu > 0$ before returning to the left-half-plane, approximately at $\mu > 0.113$, as $\mu$ increases. Figure~\ref{fig:smallmu} illustrates instability for $\mu=0.1$.

\begin{figure}
     \centering
     \begin{subfigure}[b]{0.4\textwidth}
         \centering
    \includegraphics[width=\textwidth]{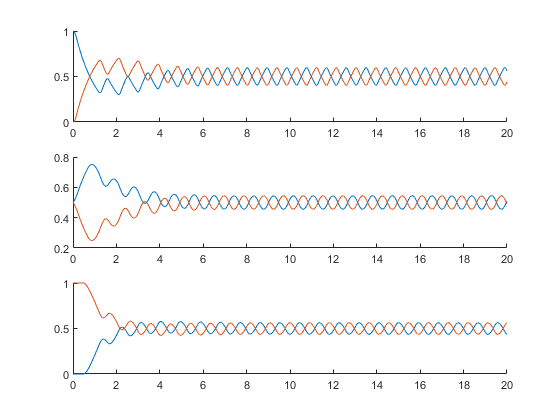}
    \caption{Player strategies in rescaled Jordan anti-coordination game: $\mu=5$.}
    \label{fig:largemu}
     \end{subfigure}
     \hfil
     \begin{subfigure}[b]{0.4\textwidth}
         \centering
    \includegraphics[width=\textwidth]{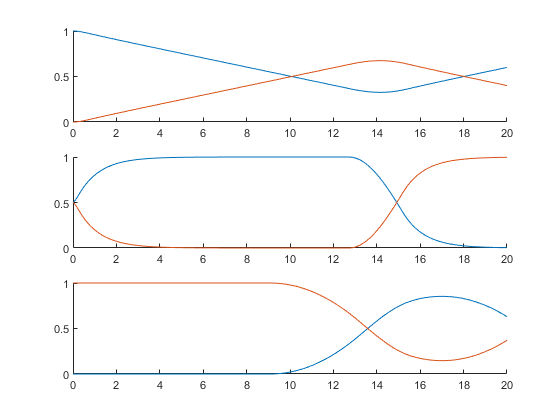}
    \caption{Player strategies in rescaled Jordan anti-coordination game: $\mu=0.1$.}
    \label{fig:smallmu}
     \end{subfigure}
    \caption{Rescaled Jordan anti-coordination game}
        \label{fig:RescaledJordan}
\end{figure}

\subsection{Stabilization of Mixed-Strategy NE in Coordination Games}
\label{sec:coordinationexm}
Next, we illustrate the implications of the inherent instability of the learning dynamics through the simple (identical interest) coordination game:
\begin{align*}
u_1(x_1,x_2) &= x_1\tr \Pmatrix{1&0\\0& 1} x_2\\
u_2(x_2,x_1) &= x_2\tr \Pmatrix{1 &0\\0& 1} x_1.
\end{align*}

Figure~\ref{fig:gpCoordination} illustrates the vector-field associated with fixed-order gradient play. There are 2 pure strategy NE that are stable, and one completely mixed-strategy NE that is unstable. 

\begin{figure}
\centering
\includegraphics[width=0.5\textwidth]{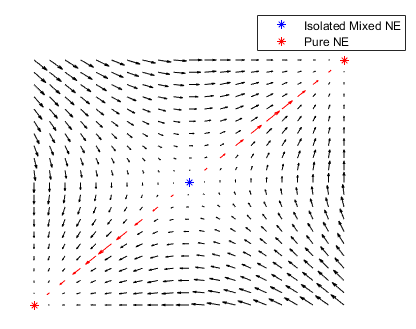}
\caption{Vector-field of fixed order gradient play.}\label{fig:gpCoordination}
\end{figure}

According to Proposition~\ref{prop:strong}, to stabilize the mixed NE at $(1/2,1/2)$, at least one player must use inherently unstable learning dynamics. We let $\xi_i\in\mathbb{R}$ for both players, and we consider the following set of parameters for higher-order gradient play: $E_1=\lambda$, $F_1=-2\lambda$, $G_1=\gamma\lambda$, $H_1=-\gamma\lambda$, $E_2=-\lambda_2$, $F_2=\lambda_2$, $G_2=-\gamma_2\lambda_2$, and $H_2=\gamma_2\lambda_2$. The numerical values are $\lambda=0.5$, $\gamma=20$, $\lambda_2=50$, and $\gamma_2=1$. Figure~\ref{fig:coordination} illustrates convergence to the mixed-strategy NE of this coordination game.

Obviously, the dynamics of $\xi_1$ are inherently unstable. Suppose we break the feedback loop and use $p^*=\Pmatrix{0\\1}$ as the input to player 1's dynamics. The response of player 1 to such input is illustrated in Figure~\ref{fig:inherent}. We see that $G_1>0$ is a scalar, and so $\xi_1$ grows without bound. The strategy $x_1$, which is projected to the simplex, converges to $\Pmatrix{1\\0}$, which is not a best response to the input payoff vector $\Pmatrix{0\\ 1}$.

\begin{figure}
     \centering
     \begin{subfigure}[t]{0.4\textwidth}
         \centering
    \includegraphics[width=\textwidth]{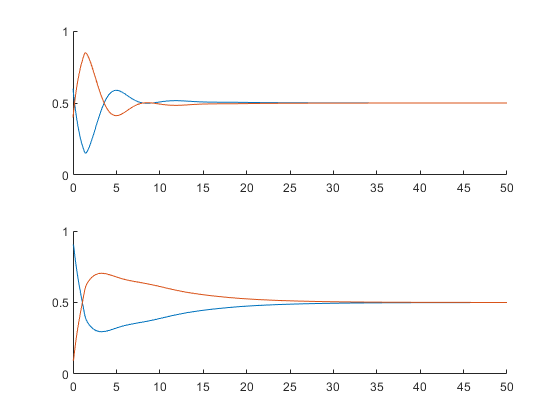}
    \caption{Coordination game: Stabilization of the mixed-strategy NE.}
    \label{fig:coordination}
     \end{subfigure}
     \hfil
     \begin{subfigure}[t]{0.4\textwidth}
         \centering
    \includegraphics[width=\textwidth]{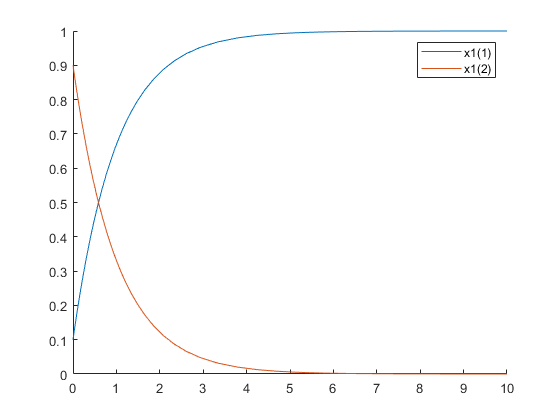}
    \caption{Inherently unstable higher-order dynamics do not converge to best response of constant payoff $p^*=\Pmatrix{0\\1}$.}
    \label{fig:inherent}
     \end{subfigure}
    \caption{Stabilizing the mixed-strategy equilibrium of a coordination game and its consequences.}
\end{figure}

\section{Concluding Remarks}

To recap, we studied the role of higher-order gradient play with linear higher-order dynamics. We showed that for any game with an isolated completely mixed-strategy NE, there exist higher-order gradient play dynamics that lead to that NE, both for the original game and for nearby games. On the other hand, we showed that for any higher-order gradient play dynamics, the dynamics do not lead to NE for a suitably rescaled anti-coordination game. We also provided an argument against dynamics that lead to 
the mixed-strategy NE in a coordination game, showing they are not reasonable. Since coordination games have pure equilibria, an interesting question remains whether one can construct such an example (where strong stabilization is impossible) with a game that has a unique mixed-strategy NE and no pure equilibria.

Regarding the higher-order gradient play dynamics that lead to NE, the interpretation of the results herein should not be that these dynamics are either a descriptive model of learning or a prescriptive recommendation for computation. Rather, the results are a contribution towards delineating what is possible or impossible in multi-agent learning, and in that sense, they may be seen as complement to the contributions in \cite{hart2003uncoupled}. Namely, dynamics being uncoupled is not a barrier to converging to mixed-strategy NE when allowing higher-order learning. 

More generally, the present results open new questions related to the discussion in \cite{hart2013simple} on what constitutes ``natural'' learning dynamics. In the case of anticipatory higher-order learning, there is clear interpretation of the effect of higher-order terms. It is unclear how to interpret general higher-order dynamics. Furthermore, the results herein regarding inherent instability of higher-order dynamics that converge to the mixed-strategy NE of a coordination game suggests that higher-order learning can be ``unnatural''. Possible restrictions on dynamics, in addition to being uncoupled, could include having no asymptotic regret; maintaining qualitative behavior in the face of strategically equivalent games (cf., Section~\ref{sec:muDiscussion}); or having an interpretable relationship between payoff streams and strategic evolutions such as ``passivity'', which generalizes and extends the notion of contractive games to contractive learning dynamics (e.g., \cite{fox2013population, arcak2021dissipativity, pavel2022dissipativity}). 

\appendix

\section{Background on Linear Systems}

Here, we review some standard background material for linear dynamics systems. There are several references (e.g., \cite{hespanha18linear,rugh1996linear}) with more detailed exposition.

\subsection{Notation}

The partitioned matrix
$$\ssSys{A}{B}{C}{D}$$
represents a general linear dynamical system
\begin{subequations}\label{eq:ABCD}
\begin{align}
\dot{x} &= Ax + Bu,\quad x(0) = x_o\\ 
y &= Cx + Du
\end{align}
\end{subequations}
whose solution is
$$y(t) = Ce^{At}x_o + \int_0^t  Ce^{A(t-\tau)}B u(\tau)\thinspace d\tau + Du(t).$$
The variables $x$, $u$, and $y$ here are used temporarily as generic placeholders (and will
play a different role in the ensuing discussion). The notation $P\sim \ssSys{A}{B}{C}{D}$ assigns the label $P$ to the dynamics (\ref{eq:ABCD}).

\subsection{Washout filters} \label{app:washout}

In the special case of
$$P\sim\ssSys{-I}{I}{-I}{I}$$
it is straightforward to show that if
$$\lim_{t\rightarrow \infty} u(t)= u^*,$$
then
$$\lim_{t\rightarrow \infty} y(t) = 0.$$
Linear systems with this property are known as ``washout filters''. For an extended discussion, see \cite{hassouneh2004washout}.


\subsection{Stability and Stabilization} 
\label{app:stabilization}
A matrix, $M$, is a \textit{stability matrix} if,  for every eigenvalue, $\lambda$, of $M$, $\mathbf{Re}[\lambda] < 0$.

The linear system, $P\sim \ssSys{A}{B}{C}{D}$ is \textit{stable} if $A$ is a stability matrix.

The linear system $K\sim \ssSys{E}{F}{G}{H}$ \textit{stabilizes} $P$ if the combined linear dynamics
\begin{align*}
\dot{x} &= Ax + Bu\\
\dot{\xi} &= E\xi + Fy\\
u &= G\xi + Hy\\
y &= Cx
\end{align*}
or
$$\Pmatrix{\dot{x}\\ \dot{\xi}} = \Pmatrix{A+BHC&BG\\FC&E}\Pmatrix{x\\ \xi}$$
are stable.

The following conditions are necessary and sufficient for the existence of such a $K$.
For all complex $\lambda$ with $\mathbf{Re}[\lambda]\ge 0$:
\begin{itemize}
\item The pair $(A,B)$ is \textit{stabilizable}: $\Pmatrix{A - \lambda I&B} = n$ has full row rank.
\item The pair $(A,C)$ is \textit{detectible}:   $\Pmatrix{C\\ A - \lambda I } = n$ has full column rank.
\end{itemize}

The linear system $K$ \textit{strongly stabilizes} $P$ if (i) $K$ stabilizes $P$ and (ii) $E$ is a stability matrix. Necessary and sufficient condition for the existence of a strongly stabilizing $K$ are presented in \cite{youla1974singleloop}. 

\subsection{Decentralized Stabilization} 

Let $P\sim\ssSys {A}{B}{C}{0}$, with $A$ having dimensions $n\times n$, have the structure
$$B = \Pmatrix{B_1&...&B_k} \And C = \Pmatrix{C_1\\ \vdots \\ C_k}$$
for some integer $k$. Suppose there exist linear systems
$$K_1 \sim \ssSys{E_1}{F_1}{G_1}{H_1}, ..., K_k \sim \ssSys{E_k}{F_k}{G_k}{H_k}$$
such that the combined linear dynamics
\begin{align*}
\dot{x} &= Ax + \Pmatrix{B_1&...&B_k} \Pmatrix{u_1\\ \vdots\\ u_k}\\
\dot{\xi}_1 &= E_1 \xi_1 + F_1 y_1\\
&\quad\vdots\\
\dot{\xi}_k &= E_k \xi_k + F_k y_k\\
u_1 &= G_1 \xi_1 + H_1 y_1\\
&\quad\vdots\\
u_k &= G_k \xi_k + H_k y_k\\
y_1 &= C_1 x\\
&\quad\vdots\\
y_k &= C_k x
\end{align*}
are stable. Then the $(K_1,...,K_k)$ achieve \textit{decentralized stabilization} of $P$. The previous conditions of $(A,B)$ 
stabilizable and $(A,C)$ detectable are necessary, but not sufficient, conditions for decentralized stabilization.

The following theorem from \cite{davison1990decentralized} provides necessary and sufficient conditions for decentralized stabilization. 
First, for any partition $Q\cup R = \theset{1,2,...,k}$
define $B\vert^Q$ as the matrix formed by extracting the block columns of $B = \Pmatrix{B_1&...&B_k}$ with indices in $Q$, i.e.,
$$B\vert^Q = \Pmatrix{B_{q_1}&...&B_{q_\magn{Q}}}$$
with $\theset{q_1,...,q_\magn{Q}} = Q$. Likewise, define $C\vert_R$ as the matrix 
formed from the block rows of $C = \Pmatrix{C_1\\ \vdots\\ C_k}$, i.e.,
$$C\vert_R = \Pmatrix{C_{r_1}\\ \vdots\\ C_{r_\magn{R}}}$$
with $\theset{r_1,...,r_\magn{R}} = R$.

\begin{Theorem}[\cite{davison1990decentralized}, Theorem 3]\label{thm:davison} 
There exist $(K_1,...,K_k)$ that achieve decentralized
stabilization of $P$ if and only if
$$\mathbf{rank}\Pmatrix{A - \lambda I&B\vert^Q\\ C\vert_R&0} = n$$
for all complex $\lambda$ with $\mathbf{Re}[\lambda]\ge 0$ and all partitions, $Q\cup R = \theset{1,2,...,k}$.
\end{Theorem}

The above rank condition must hold for \textit{all} partitions $Q\cup R = \theset{1,...,k}$. If $R=\emptyset$, one recovers the rank condition for (centralized) stabilizability. Likewise, 
$Q=\emptyset$ results in the rank condition for detectability.

\section{Proof of Theorem~\ref{thm:stabilize}}\label{app:proof}

We will examine the conditions of Theorem~\ref{thm:davison} on the system (\ref{eq:decentralized})--(\ref{eq:mathcalABC}). We need to inspect the rank of
$$\Pmatrix{\mathcal{A} - \lambda I& \mathcal{B}\vert^Q\\ \mathcal{C}\vert_R & 0}$$
for all partitions $Q\cup R = \theset{1,2,...,n}$.
Note that the partitions of either $Q=\emptyset$ or $R=\emptyset$ are already covered by Proposition~\ref{prop:fullrank}.

First, note that
$$\Pmatrix{\mathcal{A} - \lambda I& \mathcal{B}\vert^Q\\ \mathcal{C}\vert_R & 0} = 
\Pmatrix{\mathcal{M}-\lambda I&0&\Pmatrix{\mathcal{E}_{q_1} & ... & \mathcal{E}_{q_\magn{Q}}}\\
\mathcal{M}&-(\lambda+1)I&0\\
\mathcal{M}\vert_R&-I\vert_R&0}.$$
Let $\mathcal{M}$ be an $\ell\times \ell$ matrix. Then
$$\ell = \sum_{i=1}^n  (k_i - 1).$$
We need the rank of the above matrix to be $2\ell$ for all $\lambda$ with $\mathbf{Re}[\lambda] \ge 0$. Because of the presence of $\mathcal{A} - \lambda I$, loss of rank below $2\ell$ is only possible at eigenvalues of $\mathcal{A}$. Since we are only concerned with $\mathbf{Re}[\lambda] \ge 0$, we focus on eigenvalues of $\mathcal{M}$ (which excludes $\lambda = 0$ by hypothesis, since $\mathcal{M}$ is non-singular). 

Without affecting the rank, we can multiply the bottom block row by $-(\lambda+1)$ and add the middle block rows corresponding $R$ to the rescaled bottom block row to get
$$\mathbf{rank}\Pmatrix{\mathcal{A} - \lambda I& \mathcal{B}\vert^Q\\ \mathcal{C}\vert_R & 0} = 
\mathbf{rank} \Pmatrix{\mathcal{M}-\lambda I&0&\Pmatrix{\mathcal{E}_{q_1} & ... & \mathcal{E}_{q_\magn{Q}}}\\
\mathcal{M}&-(\lambda+1)I&0\\
-\lambda \mathcal{M}\vert_R&0&0}.$$
Switching the top and bottom block rows results in
$$\Pmatrix{-\lambda \mathcal{M}\vert_R&0&0\\
\mathcal{M}&-(\lambda+1)I&0\\
\mathcal{M}-\lambda I&0&\Pmatrix{\mathcal{E}_{q_1} & ... & \mathcal{E}_{q_\magn{Q}}}}.$$

We can now exploit the block triangular structure. The bottom block row provides a row rank of
$$\sum_{q\in Q} (k_q - 1),$$
The middle block row provides a row rank of $\ell$. Finally, the top block row provides a row rank of 
$$\sum_{r\in R} (k_r - 1).$$
The last assertion is because $\mathcal{M}$ is non-singular, by hypothesis, and therefore it has linearly independent rows.  Since $Q\cup R = \theset{1,...n}$, we have the desired row rank of $2\ell$.

\bibliographystyle{ieeetr}
\bibliography{Allcitations}
\end{document}